\documentclass[conference]{IEEEtran}
\IEEEoverridecommandlockouts
\usepackage{cite}
\usepackage{amsmath,amssymb,amsfonts}
\usepackage{graphicx}
\usepackage{textcomp}

\usepackage{amsmath}
\usepackage{graphicx}
\usepackage{algorithm}
\usepackage{algpseudocode}
\usepackage{caption}
\usepackage{tfrupee}

\def\BibTeX{{\rm B\kern-.05em{\sc i\kern-.025em b}\kern-.08em
    T\kern-.1667em\lower.7ex\hbox{E}\kern-.125emX}}

\algnewcommand{\LeftComment}[1]{\hskip1em$\triangleright$ #1}
\algnewcommand{\LComment}[1]{\Statex \hskip15em \(\triangleright\) #1}
\algnewcommand{\LtwoComment}[1]{\Statex \hskip17.3em \(\triangleright\) #1}
\algnewcommand{\LefComment}[1]{\Statex \hskip14em$\triangleright$ #1}

\newtheorem{theorem}{Theorem}[section]
\newtheorem{lemma}[theorem]{Lemma}

\newtheorem{corollary}[theorem]{Corollary}

\newenvironment{proof}[1][Proof]{\begin{trivlist}
\item[\hskip \labelsep {\bfseries #1}]}{\end{trivlist}}

\newcommand{\qed}{\nobreak \ifvmode \relax \else
      \ifdim\lastskip<1.5em \hskip-\lastskip
      \hskip1.5em plus0em minus0.5em \fi \nobreak
      \vrule height0.75em width0.5em depth0.25em\fi}

\long\def\/*#1*/{}

\begin{document}
%
\title{An algorithmic approach to handle circular trading in commercial taxing system}

\/*
\author{\IEEEauthorblockN{1\textsuperscript{st} Jithin Mathews}
\IEEEauthorblockA{\textit{Department of Computer Science} \\
\textit{Indian Institute of Technology}\\
India\\
Email: {cs15resch11004@iith.ac.in}}
\and
\IEEEauthorblockN{2\textsuperscript{nd} Priya Mehta}
\IEEEauthorblockA{\textit{Department of Computer Science} \\
\textit{Indian Institute of Technology Hyderabad}\\
India\\
Email: {cs15resch11007@iith.ac.in}}
\and
\IEEEauthorblockN{3\textsuperscript{rd} Ch. Sobhan Babu}
\IEEEauthorblockA{\textit{Department of Computer Science} \\
\textit{Indian Institute of Technology Hyderabad}\\
India\\
Email: {sobhan@iith.ac.in}}
}
*/

\author{\IEEEauthorblockN{Jithin Mathews}
\IEEEauthorblockA{\textit{Department of Computer Science} \\
\textit{Indian Institute of Technology}\\
India\\
Email: {cs15resch11004@iith.ac.in}}
\and
\IEEEauthorblockN{Priya Mehta}
\IEEEauthorblockA{\textit{Department of Computer Science} \\
\textit{Indian Institute of Technology Hyderabad}\\
India\\
Email: {cs15resch11007@iith.ac.in}}
\and
\IEEEauthorblockN{S.V. Kasi Visweswara Rao}
\IEEEauthorblockA{\textit{Department of Commercial Taxes}\\
\textit{Government of Telangana}\\
India\\
Email: {svkasivrao@gmail.com}}\\
\and
\IEEEauthorblockN{Ch. Sobhan Babu}
\IEEEauthorblockA{\textit{Department of Computer Science} \\
\textit{Indian Institute of Technology Hyderabad}\\
India\\
Email: {sobhan@iith.ac.in}}
}

\maketitle

\begin{abstract}
Tax manipulation comes in a variety of forms with different motivations and of varying complexities. In this paper, we deal with a specific technique used by tax-evaders known as circular trading. In particular, we define algorithms for the detection and analysis of circular trade. To achieve this, we have modelled the whole system as a directed graph with the actors being vertices and the transactions among them as directed edges. We illustrate the results obtained after running the proposed algorithm on the commercial tax dataset of the government of Telangana, India, which contains the transaction details of a set of participants involved in a known circular trade.
\end{abstract}

\begin{IEEEkeywords}
data mining, bigdata analytics, social network analysis, circular trading, forensic accounting, value added tax.
\end{IEEEkeywords}

\section{Introduction}

Fraudulent activity, unfortunately, is inherent in our society from time immemorial. It is primarily motivated by the unscrupulous desire of people to make personal benefits by exploiting the loopholes in the existing laws in a system. Certain types of fraudulent activities are easier to identify and scrutinize. On the other hand, there are fraudulent methods that are extremely difficult to track down due to the complexity of the processes involved in handling them. In \cite{1ref}, Van Vlasselaer et al. gives a formal, concise and complete definition of $`$fraud':

$``$Fraud is an uncommon, well-considered,
imperceptibly concealed, time-evolving and often
carefully organized crime which appears in many
types of forms.$"$

In this paper, we propose a systematic technique using social network analysis to handle a complicated type of financial fraud, which is widely rampant in the commercial taxing system, known as $\it{circular\,trading}$. It is committed by business entities with the intention of evading tax, which they are liable to pay to the government. $\it{Circular\,trading}$ is a theft of Value Added Tax (or VAT) from the government by a business entity by creating fictitious business firms and diligently organizes with them to manipulate the financial information submitted in their commercial tax return filing. It is similar to the infamous $\it{carousal\,fraud}$ \cite{2ref}, which is a comparatively less sophisticated method, used by fraudsters for tax-evasion. $\it{Bill\,trading}$ \cite{2.1ref} is another technique used in tax-evasion where a dealer sells some goods to another dealer without raising an invoice, but collects the tax from him. The former dealer then issues fake invoice to a third dealer who uses it to minimize his tax liability. Note that for conducting the proposed research work we have used the commercial tax data set shared by the Telangana state government, India. 

In VAT system, when a business dealer, say dealer $B$, purchases some goods from another dealer, say dealer $A$, dealer $B$ is liable to pay a certain amount of tax on the purchased goods to dealer $A$ and let us call it as the $input\,tax$ paid by dealer $B$ to dealer $A$ on the business transaction. Similarly, when dealer $B$ sells these goods to another dealer, say dealer $C$, dealer $B$ will receive a certain amount of tax on the sold goods from dealer $C$ and let us call it as the $output\,tax$ received by dealer $B$ from dealer $C$ on the business transaction. In this case, the amount of tax received by the government from dealer $B$ is equal to the difference between the $output\,tax$ received by $B$ and the $input\,tax$ paid by $B$. In other words, $tax\,payable = (output\,tax\,received - input\,tax\,paid)$.

This formula is universal for any business dealer. However, when this difference becomes a negative value, $i.e.$, when the $input\,tax$ paid becomes greater than the $output\,tax$ received, the dealer will receive Credit Carry Forward (or CCF) \cite{3ref}, which (s)he can claim from the government or can use it against paying tax in the future. Note that through out the paper cash is represented in Indian currency $``Rupees"$ ($Rs.$ or \rupee). In Figure $1$, we pictorially illustrate the flow of money in a value added taxing system. Here, the producer, who makes raw materials, sells them to a manufacturer for \rupee~1200 imposing 10\% of tax and thereby collecting \rupee~120 in tax. Since producer does not have any input tax, the $tax\,payable$ is \rupee~120 and he pays it to the government. The manufacturer processes the raw materials, makes it into a product and sells it to a retailer for a higher price. Here he collects a tax of \rupee~180 from the retailer. The amount of tax that the manufacturer needs to pay to the government as a result of the previously mentioned value addition is, $tax\,payable = (180 - 120) = $ \rupee~60. Finally, the retailer adds more value to the product, like, the packing of the product, and sells it to a consumer for a higher price by collecting a tax of \rupee~200. In this case, $tax\,payable$ by the retailer to the government is $(200 - 180) = $ \rupee~20. Hence a total of \rupee~200 (= \rupee~120 + \rupee~60 + \rupee~20) is collected by the government from different stages of this transaction.

\begin{figure}[h!]
\begin{center}
  \caption{Flow of money in VAT system}
\includegraphics[scale=0.3]{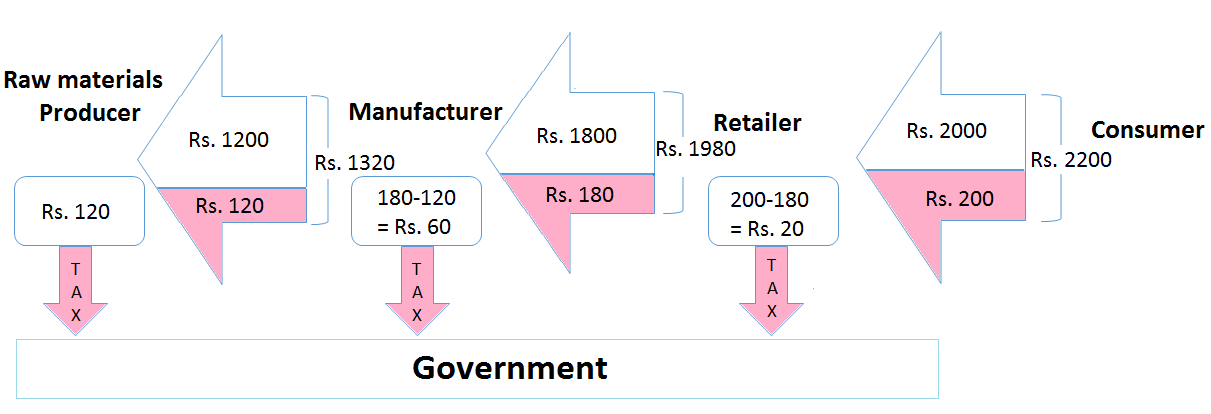}
\label{graph_table1}
\end{center}
\end{figure} 

\subsection{$\it{Circular\,trading}$}

The primary motivation for $\it{circular\,trading}$ is to hide malicious sales and(or) purchases information from the tax enforcement officers. This is done by superimposing those transactions by carefully fabricated transactions, which we call as illegitimate transactions throughout the paper. The classical theme in such an evasion is described in the following steps:

$Step\,1.$ Dealer $A$ would purposefully omit some of his/her sales and purchases information in the tax returns. These malicious tax-return information will result in the reduction of the dealer's $tax\,payable$ and he/she ends up paying less tax to the government. However, this cannot continue for longtime since the dealer's financial growth may not be in proportion to the amount of tax (s)he pays and consequently becomes more likely to get caught.

$Step\,2.$ Guided with the intention to hide the manipulation in his/her tax returns, dealer $A$ will create a few fictitious dealers using the personal identification details of his/her trusted acquaintances.

$Step\,3.$ At this stage, dealer $A$ will fabricate numerous sales and purchases information between himself and the fictitious dealers by making sure that the fabricated sales and purchases information are liable to a negligible amount of tax. The tax payable on these illegitimate transactions is almost zero since they amount to almost zero value addition.

Hence, dealer $A$ ingeniously manages to camouflage into the nexus of fictitious dealers that (s)he has created. In fact, this helps the dealer to successfully suppress his/her sales and purchases information without getting into the hands of tax enforcement officers.

Despite of the carefully orchestrated manipulations, the dealer engaged in $\it{circular\,trading}$ cannot avoid giving rise to undesired patterns in the flow of transactions. In this paper, we exploit this facet of the manipulated tax returns. One can easily observe that the manipulation, as defined in the last three steps, will result in the formation of flow of goods in a circular manner. For example, in $Step\,3$, which is illustrated in Figure $2$, dealer $A$ seems to sell some goods to another dealer, say to dealer $B$, and dealer $B$ seems to sell the same kind of goods to dealer $C$, and finally dealer $A$ purchases the same kind of goods from dealer $C$, hence completing the cycle. Note that the $Value$ of goods transferred is almost the same in all the three transactions that create the cycle. Generally, this is not a desired pattern for the flow of goods if the transactions are authentic. These cycles become much complicated to analyze with the involvement of more than 3 dealers.

\begin{figure}[h!]
\begin{center}
  \caption{Circular trading}
\includegraphics[scale=0.4]{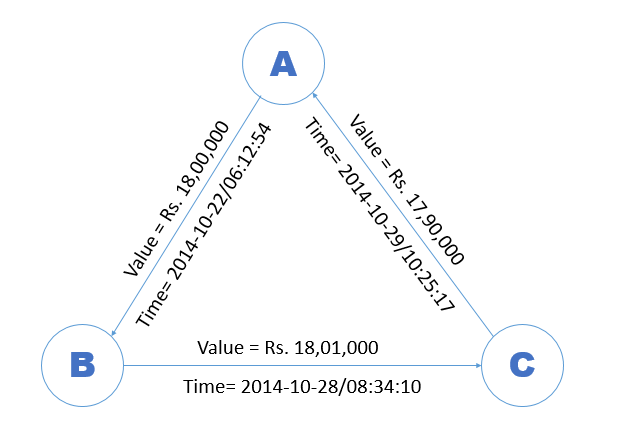}
\label{graph_table2}
\end{center}
\end{figure}

The main obstacles in identifying malicious sales transactions are the large size of the dataset, complex sequences of the illegitimate transactions and the large number of traders involved in $\it{circular\,trading}$. In this paper, we propose an algorithm to remove the illegitimate transactions, which are superimposed on the malicious sales transactions. This allows tax authorities to identify malicious transactions in an easy manner. The three steps detailed in this section make the central theme for $\it{circular\,trading}$. Dealers who commit this fraud often adds up more complexity to the problem by exploiting the way VAT system works in a multi-jurisdictional trading. However, the concept of goods circling around in a cycle or a circular fashion remains the same.

\section{Related Work}

Most of the work on $\it{circular\,trading}$ are concentrated on stock market trading. In \cite{4ref}, \cite{5ref}, \cite{6ref}, \cite{7ref} and \cite{8ref} the authors have investigated on $\it{circular\,trading}$ and other related collusion techniques used in stock market trading. A brief overview on some of these techniques is given below.

In \cite{4ref}, a graph clustering algorithm is devised for detecting collusion sets in stock markets. A novel feature of this approach is the use of Dempster–Schafer theory of evidence to combine the candidate collusion sets. In \cite{7ref}, a method is proposed to detect the potential collusive cliques involved in an instrument of future markets. In \cite{8ref}, the authors introduced complicity functions, which are capable of identifying the intermediaries in a group of actors, avoiding core elements that have nothing to do with the group. 

To the best of our knowledge, no formal techniques have been devised to handle $\it{circular\,trading}$ in taxation system. However, there are few existing works done to fight tax-evasion. In \cite{9ref}, the authors presented a technique relying upon statistical methods for detecting Value Added Tax evasion by Kazakhstani legal entities. Starting from feature-selection, they performed an initial exploratory data analysis using Kohonen self-organizing maps, which allowed them to make basic assumptions on the nature of tax compliant companies. Then they selected a statistical model and proposed an algorithm to estimate its parameters in an unsupervised manner. In \cite{10ref}, the authors presented a case study of a pilot project developed to evaluate the use of data mining in audit selection for the Minnesota Department of Revenue.

In Section III, we reduce the problem in hand to a graph-theoretical problem of \emph{`Deleting cycles from a weighted directed graph in such a way that the difference between the weights of the highest-weighted-edge and the lowest-weighted-edge in the cycle is minimized.'} Several approaches are available in the literature to detect and delete cycles from a directed graph. We observed that most of these works are motivated from some real world problems, just as in our case it is to scrutinize tax-evasion done by $\it{circular\,trading}$. To the best of our knowledge, no work has been proposed in the past for deleting cycles from a graph, which is similar to the technique we describe in this paper. Traditional methods rely on depth-first searches (DFS) \cite{10.1ref}, exploiting the fact that a graph has a cycle iff DFS finds a so-called back edge. Recent works on this topic include distributed algorithms, which aim to maintain an acyclic graph when new links are added to an initially acyclic graph \cite{11ref}. In \cite{12ref}, the authors introduced a new problem: cycle detection and removal with vertex priority. It proposes a multi-threading iterative algorithm to solve this problem for large-scale graphs on personal computers. In \cite{13ref} authors considered the problem of detecting a cycle in a directed graph that grows by arc insertions, and the related problems of maintaining a topological order and the strongly connected components in such graphs.

\section{Problem Definition}

In this section, we define the problem formally using graph theoretic terminologies and give a brief overview on the methodology used for handling the same. A thorough description of the algorithm along with its correctness and time complexity is given in the next section.

Table I shows a snapshot of the dataset used. `ID' is the unique identity number of a dealer. `Seller's ID' and `Buyer's ID' shows the direction of the flow of goods, `Time' gives the exact time of the transaction including the date, and the variable `Value' is the amount of tax paid by the buyer to seller. For example, the second row in Table $1$ can be interpreted as a dealer with ID $a$ selling goods to a dealer with ID $b$ on January $14^{th}$ of $2015$ at local time $1$:$01$:$54$ $pm$ and the buyer, dealer with ID $b$, gives a tax of \rupee~$15,000$ to the seller.

We denote the system of all transactions using a weighted directed graph $G=(V,E)$. Here $V$, which is the vertex set, is a set containing the ID's of all dealers in the transactions. A transaction is defined using a weighted directed edge, and the set of all these edges are denoted by $E$. The weight on any edge is a $2$-tuple of its corresponding `Value' and `Time' attribute values, $(Value,Time)$. So the second row in Table $1$ can be translated as a directed edge $\vec{ba}$ with weight $(15000,2015/01/14/ 13$:$01$:$54)$. Note that graph $G$ may contain multiple edges but no self loops. All multiple edges can be uniquely identified using the `Time' attribute in its weight since we can safely assume that no two similar transactions between two dealers occur exactly at the same time, $i.e.,$ a dealer $A$ cannot make two separate sales transactions to a dealer $B$ at the same time. The same applies to two purchase transactions by dealer $A$ from dealer $B$.

\begin{table}[ht]
\caption{Sales transactions dataset}
\begin{center}
\begin{tabular}{|c|c|c|c|c|}
\hline
Serial.No. & Seller's ID & Buyer's ID & Time & Value in \rupee\\
\hline
1 & m & n & 2015/01/14/10:30:44 & 10000\\
\hline
2 & a & b & 2015/01/14/13:01:54 & 15000\\
\hline
3 & x & y  & 2015/01/15/09:02:52 & 12000\\
\hline
4 & y & m & 2015/01/15/10:09:11 & 14000\\
\hline
5 & b & k & 2015/01/16/10:10:10 & 10000\\
\hline
\end{tabular}
\end{center}
\label{salestransactions}
\end{table}

As mentioned in the last Section I, $\it{circular\,trading}$ results in the formation of undesired flow of goods in a circular fashion, which we call as cycles in graph theoretic terms. The problem of removing these cycles is important as the tax authorities can easily detect the malicious transactions once the cycles are removed. Note that deleting an edge from a cycle results in the absence of that cycle from the graph. The order in which we delete cycles is significant since different order of edge deletion produces different directed acyclic graphs ($DAG$s) at the end. This is due to the simple fact that different cycles may share one or more edges among each other.

\begin{figure}[h!]
\begin{center}
  \caption{Cycle deletion}
\includegraphics[scale=0.52]{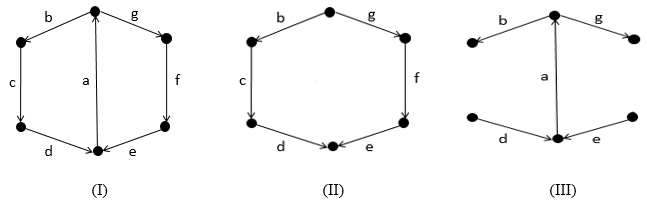}
\label{graph_table3}
\end{center}
\end{figure} 

For example, as illustrated in Figure $3$, if a graph (given in (I)) contains two cycles that share a common edge $a$, $viz.$ $(a,b,c,d,a)$ and $(a,g,f,e,a)$, deleting edge $a$ results in the formation of a different $DAG$ (as given in (II)) from the $DAG$ formed by deleting one edge each from each cycle that is not edge $e$, as given in (III). Hence, we chose an ordering technique for edge deletions using the ideas in Observation $1$, which was given by the taxation authorities.

$\bf{Observation\,1.}$
\emph{In $\it{circular\,trading}$ a dealer fabricates sales and purchases information between himself and the fictitious dealers such that the $\it{input\,tax}$ and the $\it{output\,tax}$ due to these illegitimate transactions are almost the same, ($i.e.$, $\it{tax\,payable}$ on the illegitimate transactions are nullified).}

The $Value$ parameter of the three transactions shown in Figure $2$ of Section I illustrates $Observation\,1$. A careful study of this observation naturally results in deleting cycles in the following particular order:

\emph{`Delete cycles in such a way that the difference between the tax values of the highest-tax-valued-edge in the cycle, (where, `Tax value' is the second element in the $2$-tuple denoting the weight of an edge), and the lowest-tax-valued-edge in the cycle is minimized.'}

This order of cycle deletion makes sense, since, for a dealer in any illegitimate cycle the tax he pays on an illegitimate purchase transaction should be almost the same as the tax he collects on an illegitimate sales transaction. This makes the net tax liability due to illegitimate transactions nearly zero.

\section{Design and analysis of the algorithm to delete cycles from a weighted directed graph in a particular order}

The entire technique of deleting cycles is covered in algorithms $1$, $2$ and $3$. Algorithm $1$ invokes a function defined in algorithm $2$, which in turn invokes a function defined in algorithm $3$. We give the complete algorithm, a brief overview of the same, along with its proof of correctness and time complexity analysis in this section. First of all, let us define few useful terminologies:\

\begin{itemize}

\item If there exist multiple edges from vertex $x$ to vertex $y$, then $max(e_{xy})$ denotes the edge with the maximum $Value$ among all edges directed from $x$ to $y$.

\item $\it{Critical\,edge}$ of a path $P$ or a cycle $C$ in a graph is an edge in the corresponding path or the cycle with the minimum $Value$. We denote it by $\gamma_{P}$ or $\gamma_{C}$, respectively.

\item $\it{Maxflow\,path}$ from a vertex $x$ to a vertex $y$ in a graph is the path with the $Value$ of its $\it{Critical\,edge}$  being the maximum among all the paths from vertex $x$ to vertex $y$. We denote it by $\mu_{xy}$. Note that vertices $x$ and $y$ cannot be the same, in which case we have a cycle and not a path. Hence, in such cases where $x=y$, we consider the $``$path$"$, say $\mu_{xx}$ (or $\mu_{yy}$), to be an unreachable path with the $Value$ of its $\it{Critical\,edge}$ equals $+\infty$, $i.e.$, $Value(\gamma_{\mu_{xx}}) = +\infty$.

\item $\it{Flow\,value}$ of a path $P$ or a cycle $C$ in a graph is the difference between the $Value$ of the maximum-valued-edge and the $Value$ of the minimum-valued-edge (minimum-valued-edge is same as $\it{Critical\,edge}$) in the path or the cycle. We denote it by $\phi_{P}$ or $\phi_{C}$, respectively.
\end{itemize}

\subsection{Algorithm overview}

\begin{itemize}

\item In algorithm $3$, we find the $\it{Maxflow\,path}$ between two vertices, from vertex $v$ to vertex $u$, in the input graph $G^{''}$, $i.e.$ the path $\mu_{vu}$, and returns it to algorithm $2$.\\

Here we describe the main steps involved in algorithm $3$. We start by initializing two vectors, $viz.,$ dist$[]$ and parent$[]$. Vector parent$[]$ is initialized to null for all the vertices, while vector dist$[]$ is initialized to $-\infty$ for all vertices except for the source vertex $v$, dist$[v] = +\infty$. Then all vertices in the vertex set of the input graph $G^{''}$ is inserted into a priority queue (max-heap) based on their dist$[]$ values. Then, we delete the vertex in the queue with the highest dist$[]$ value, and during each such deletion the dist$[]$ and parent$[]$ vectors of the deleted vertex's outgoing-neighbors (vertex $n$ is an outgoing-neighbor of a vertex $x$ if the graph contains an edge directed from vertex $x$ to vertex $n$) are updated as explained below.\\

Let us call the deleted vertex as vertex $ver$. The dist$[]$ and parent$[]$ values of any outgoing-neighbor of vertex $ver$ are updated if we find a better path from the source vertex $v$ to the corresponding outgoing-neighbor. In other words, both vectors of an outgoing-neighbor of vertex $ver$ are updated if the $Value$ of the $\it{Critical\,edge}$ in the new path from vertex $v$ to the outgoing-neighbor $via$ vertex $ver$ is greater than the current dist$[]$ value of the outgoing-neighbor. Note that vector dist$[]$ is updated with the $Value$ of the $\it{Critical\,edge}$ in the new path and vector parent$[]$ is updated with the edge between $ver$ and its outgoing-neighbor. The process of deleting vertices from the queue will continue until the queue becomes empty. Once all the vertices are deleted from the queue, for any vertex $x$ belonging to $G^{''}$, dist$[x]$ represents the $Value$ of the $\it{Critical\,edge}$ in the $\it{Maxflow\,path}$ $\mu_{vx}$ in graph $G^{''}$. The $\it{Maxflow\,path}$ $\mu_{vu}$ is returned to algorithm $2$ by backtracking from the vertices present in vector parent$[u]$ to vertex $v$.

\begin{figure}[h!]
\begin{center}
\captionsetup{justification=centering}
  \caption{Cycle formation}
\includegraphics[scale=0.7]{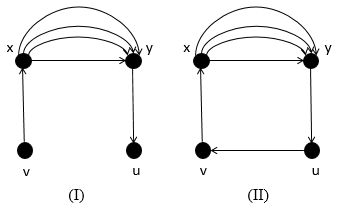}
\label{graph_table4}
\end{center}
\end{figure}

\item Algorithm $2$ takes graph $G^{'}$ and an edge $e$ as input, where edge $e$ is directed from vertex $u$ to vertex $v$. This algorithm removes a cycle $C$, which contains edge $e$, from graph $G^{'}$ such that its $\it{Flow\,value}$ $\phi_{C}$ is the minimum among the $\it{Flow\,values}$ of all the cycles containing edge $e$.\\

It is important to note that the addition of an edge can give rise to several cycles as the graph may contain multiple edges. For example, in Figure $4$, the addition of an edge from vertex $u$ to vertex $v$ in the graph given in (I) will create $4$ different cycles as shown in (II). Hence, we need to decide which cycle is to be deleted before the other. In algorithm $2$, we delete a cycle according to the order defined in Section III, $i.e.$, \emph{`Delete cycles in such a way that the difference between the tax values of the highest-tax-valued-edge and the lowest-tax-valued-edge in the cycle is minimized.'}\\

Here, we invoke algorithm $3$ using graph $G^{''}$, where $G^{''}$ is a copy of the input graph $G^{'}$, and the vertices $u$ and $v$ as parameters. Recall that algorithm $3$ returns the $\it{Maxflow\,path}$ $\mu_{vu}$ of $G^{''}$, and we store it as path $P$. After adding path $P$ to a set $S$, we delete all the edges from graph $G^{''}$ whose $Value$ is greater than or equal to the $Value$ of the maximum-valued-edge in path $P$. We continue this process until no cycles are left in $G^{''}$. At this point, set $S$ contains different paths from vertex $v$ to vertex $u$. Then we add the edge $e$, $i.e.$, the edge from vertex $u$ to vertex $v$, to each of the paths stored in set $S$. It is easy to see that, after the addition of edge $e$, set $S$ contains only cycles in it. Then we find the cycle, say cycle $P_{min}$ in set $S$, whose $\it{Flow\,value}$ $\phi_{P_{min}}$ is the minimum among all the cycles present in $S$. Finally, we remove the cycle $P_{min}$ by subtracting a value equal to $\phi_{P_{min}}$ from each of the edges in $P_{min}$, and the resultant graph is returned to algorithm $1$.

\/*
In this algorithm, we invoke algorithm $3$ using graph $G^{''}$, where $G^{''}$ is a copy of the graph $G^{'}$, and the vertices $u$ and $v$ as parameters (refer to Step $5$). Recall that algorithm $3$ returns the $\it{Maxflow\,path}$ $\mu_{vu}$ of $G^{''}$, and in Step $5$ of this algorithm we store it as path $P$. After adding path $P$ to a set named as $S$ (in Step $6$), we delete all edges from graph $G^{''}$ whose $Value$ is greater than or equal to the $Value$ of the  maximum-valued-edge in the path $P$ (refer to Step $7$). We continue this process in the while loop given in steps $4$ to $8$ until no cycles are left in $G^{''}$. Note that, once we are out of the loop, set $S$ contains different paths from vertex $v$ to vertex $u$. In Step $9$, we add the edge $e$, $i.e.$, the edge from vertex $u$ to vertex $v$, to each of the paths stored in set $S$. It is easy to see that after the addition of edge $e$ set $S$ now contains only cycles in it. In Step $10$, we find the cycle $P_{min}$ in set $S$ whose $\it{Flow\,value}$, $\phi_{P_{min}}$, is the minimum among all the cycles present in $S$. Finally, in order to remove the cycle $P_{min}$, a value equal to $\phi_{P_{min}}$ is removed from each edge of $P_{min}$ in Step $11$ and the resultant graph is returned to algorithm $1$.
*/

\item In algorithm $1$, the input graph has all its edges sorted in the increasing order of time. It invokes algorithm $2$, which deletes a cycle in the proposed order, until no cycles are left in the graph and the resulting directed acyclic graph is the desired output graph.\\

We consider the edge set of the graph as a queue and starts deleting elements from it. Note that always the least recent edge is deleted from the edge set as all its edges are arranged in chronological order. The deleted edges are then inserted into a new graph, say graph $G^{'}$, in the same order as they are deleted and whenever a cycle is detected in the new graph due to the addition of an edge, the function defined in algorithm $2$ is invoked with the graph $G^{'}$ and the most recently inserted edge as parameters. Algorithm $2$ will delete the cycle in the specific order mentioned before, and this process will continue until there exists no cycle in graph $G^{'}$. 

\end{itemize}

\subsection{Correctness of the algorithm}
In this section, we prove the correctness of the proposed algorithm.

Let vertex $v$ be the vertex deleted in some iteration of Step $5$ in algorithm $3$. One can easily verify that after the execution of Step $7$, value of the vector dist$[]$ for any outgoing-neighbor $n$ of vertex $v$ can be defined by the following formula:

\small
\begin{equation}
dist[n] = MAX\Bigg(dist[n], MIN\bigg(dist[v], Value\Big(max(e_{vn})\Big)\bigg)\Bigg)
\end{equation}
\normalsize

where $MAX()$ and $MIN()$ functions return the maximum and minimum values, respectively.

                            
\begin{lemma}
After the execution of algorithm $3$, for each $m \in V^{''}$, dist$[m]$ = $Value(\gamma_{\mu_{vm}})$, where $\gamma_{\mu_{vm}}$ represents the $\it{Critical\,edge}$ in the $\it{Maxflow\,path}$ $\mu_{vm}$.
\end{lemma}

\begin{proof}
Assume that set $O$ contains all vertices deleted from the Queue Q in Step $5$ of algorithm $3$. 
We use induction on $|O|$ to prove $Lemma\,3.1$.\

\noindent $Base\,case\,(|O|=1):$ Here we prove the lemma for the first vertex inserted into set $O$, $i.e.,$ the first vertex which is deleted from the queue Q. Initially, since dist$[v] = +\infty$ and dist$[w] = -\infty$, $\forall w \in V^{''} \setminus v$, dist$[v] >$ dist$[w]$. Therefore, $v$ is the first vertex deleted from Q in Step $5$ of algorithm $3$. Clearly, $Lemma\,3.1$ holds in this case as the $\it{Maxflow\,path}$ $\mu_{vv}$, where the source vertex and the destination vertex are the same, does not exists as per our definition of a $\it{Maxflow\,path}$ given in the beginning of this section. Hence, for the $Base\,case$, $Lemma\,3.1$ holds and dist$[v] = +\infty = Value(\gamma_{\mu_{vv}})$.

\noindent $Inductive\,hypothesis:$ Let $x$ be the last vertex added to set $O$, and assume $O' = O \cup \{x\}$. In this case, the $Inductive\,hypothesis$ states that $\forall y \in O$, dist$[y] = Value(\gamma_{\mu_{vy}})$.

\noindent $Inductive\,step:$ The inductive proof is complete if we show that dist$[x]$ = $Value(\gamma_{\mu_{vx}})$.

Let $v, a_{1}, a_{2}, \cdots a_{k}, m, n, c_{1}, c_{2}, \cdots c_{l}, x$ be the $\it{Maxflow\,path}$ $\mu_{vx}$ from vertex $v$ to vertex $x$. Assume that $v, a_{1}, a_{2}, \cdots a_{k}, m \subseteq O$ and $n \in V{''} - O$. By $Inductive\,hypothesis$, dist$[m] = Value(\gamma_{\mu_{vm}})$. Let $P$ and $P'$ denote the sub-paths $v, a_{1}, a_{2}, \cdots a_{k}, m$ and $v, a_{1}, a_{2}, \cdots a_{k}, m, n$ of the $\it{Maxflow\,path}$ $\mu_{vx}$, respectively. 
Clearly, $Value$ of the $\it{Critical\,edge}$ of the sub-path $P'$ is equal to:\\

\begin{align}
Value(\gamma_{P'}) = MIN \Big( Value \big(\gamma_{P} \big), Value \big(max(e_{mn}) \big) \Big) \\\leq MIN \Big( Value \big(\gamma_{\mu_{vm}} \big), Value \big(max(e_{mn}) \big) \Big) \\ \implies Value(\gamma_{P'}) \leq MIN \Big( dist[m], Value \big(max(e_{mn}) \big) \Big)
\end{align}

Since vertex $n$ is an outgoing-neighbor of vertex $m$, after the deletion of vertex $m$ from the queue, dist$[n]$ gets updated with the following result as derived from $Formula\,1$:
\small
\begin{equation*}
dist[n] = MAX\Bigg(dist[n], MIN\bigg(dist[m], Value\Big(max(e_{mn})\Big)\bigg)\Bigg)
\end{equation*}
\normalsize

\noindent Using $Inequality\,4$, the previous result can be rewritten as:
\small
\begin{equation*}
dist[n] \geq MAX\bigg(dist[n], Value \Big(\gamma_{P'} \Big) \bigg)
\end{equation*}
\normalsize

Consequently, dist$[n] \geq Value(\gamma_{P'})$. Since $P'$ is a sub-path of the $\it{Maxflow\,path}$ $\mu_{vx}$, $Value(\gamma_{\mu_{vx}}) \leq Value(\gamma_{P'})$. Hence, dist$[n] \geq Value(\gamma_{\mu_{vx}})$. In Step $5$ of algorithm $3$, we delete the vertex with the highest dist$[]$ value and as vertex $x$ is deleted from the queue before vertex $n$, dist$[x] \geq$ dist$[n]$ which implies dist$[x] \geq Value(\gamma_{\mu_{vx}})$.

If dist$[x] > Value(\gamma_{\mu_{vx}})$, then, the $Value$ of the $\it{Critical\,edge}$ in the path formed by following the sequence of vertices deleted from the queue starting at vertex $v$ and ending at vertex $x$, is greater than $Value(\gamma_{\mu_{vx}})$. This means that path $\mu_{vx}$ is not a $\it{Maxflow\,path}$ from vertex $v$ to vertex $x$, which contradicts our assumption. Therefore, dist$[x] \ngtr Value(\gamma_{\mu_{vx}})$ which implies dist$[x] = Value(\gamma_{\mu_{vx}})$. Hence $Lemma\,3.1$ holds true for vertex $x$.

This completes the proof of $Lemma\,3.1$.
\end{proof}

\begin{corollary}
The MAX\_MIN() function defined in algorithm $3$ finds the $\it{Maxflow\,path}$ from the source vertex $v$ to any vertex $m$ in the input graph $G^{''}$ which can be retrieved by backtracking from the vector parent$[m]$ to the vertex $v$.
\end{corollary}

\begin{proof}
The above corollary is true since the vector parent$[]$ is updated in Step $7$ of algorithm $3$ iff dist$[]$ is updated, and according to $Lemma\,3.1$ for each $m \in V^{''}$, dist$[m]$ = $Value(\gamma_{\mu_{vm}})$.
\end{proof}

\/*

\begin{corollary}
Suppose $M(P)$ denotes the maximum $Valued$ edge in a path $P$ in any graph. Then there exists no path $P^{'}$ with $Value(M(P^{'})) \geq Value(M(\mu_{xy}))$ and $\phi_{P^{'}} < \phi_{\mu_{xy}}$, where $\mu_{xy}$ is a $\it{Maxflow\,path}$ from vertex $x$ to vertex $y$ in the input graph $G^{'}$ of algorithm $2$ and $P'$ is any other path in $G^{'}$ from $x$ to $y$, and $\phi_{P^{'}}$ is the $\it{Flow\,value}$ of path $P^{'}$.
\end{corollary}

\begin{proof}
Recall that $\phi_{P} = \big(Value(M(P)) - Value(\gamma_{P})\big)$, where $\gamma_{P}$ is the $\it{Critical\,edge}$ in path $P$, which is same as the minimum-$Valued$-edge in $P$. Therefore, if $\phi_{P^{'}} < \phi_{\mu_{xy}}$ and $Value(M(P^{'})) \geq Value(M(\mu_{xy}))$, then $Value(\gamma_{P^{'}}) > Value(\gamma_{\mu_{xy}})$, which is a contradiction to our assumption that $\mu_{xy}$ is a $\it{Maxflow\,path}$ from vertex $x$ to vertex $y$. Hence the corollary is true.
\end{proof}

*/

\begin{lemma}
Let $C_1, C_2, C_3, \cdots, C_{k-1}, C_k$ be the cycles present in set $S$ ordered in the increasing order of their maximum-$Valued$-edges (we denote this ordering as $\vec{O}$) after the execution of algorithm $2$. Then, the order in which these cycles are deleted from graph $G^{''}$ in Step $7$ of algorithm $2$ is in the reverse order of ordering $\vec{O}$, $i.e.$, cycle $C_k$ is deleted at first, then cycle $C_{k-1}$, $\cdots, C_3,\,C_2$ and at last cycle $C_1$.
\end{lemma}

\begin{proof}
It is easy to observe that no two cycles in $\vec{O}$ can have the same maximum-$Valued$-edge. This is due to the fact that in Step $7$ of algorithm $2$ all edges from graph $G^{''}$ are deleted whose $Value \geq Value(M(P))$, where $P$ is the $\it{Maxflow\,path}$ from vertex $v$ to vertex $u$ in $G^{''}$ (note that, later in Step $9$, the operation $P=P\cup \{e\}$ causes this $\it{Maxflow\,path}$ $P$ to become a cycle). Following the definition of $\vec{O}$, $Value(M(C_{1})) < Value(M(C_{2})) < Value(M(C_{3})) < \cdots < Value(M(C_{k-1})) < Value(M(C_{k}))$, where $M(C)$ denotes the maximum-$Valued$-edge in a given cycle $C$. Suppose that $Lemma\,3.3$ is wrong, then, there exists two cycles $C_i$ and $C_l$ such that cycle $C_i$ comes before cycle $C_l$ in $\vec{O}$ and $C_i$ is deleted before $C_l$ from graph $G^{''}$ in Step $7$ of algorithm $2$. Since $C_i$ is deleted before $C_l$, $Value(M(C_l)) < Value(M(C_i))$. Also, since $C_i$ comes before $C_l$ in $\vec{O}$, $Value(M(C_i)) < Value(M(C_l))$ and hence we reach a contradiction. Therefore, $Lemma\,3.3$ holds true.
\end{proof}

\begin{lemma}
Algorithm $2$ deletes a cycle from graph $G^{'}$ with the minimum $\it{Flow\,value}$ among all the cycles in $G^{'}$.\end{lemma}

\begin{proof}
Let $C_{min}$ denote the cycle with the minimum $\it{Flow\,value}$ among all cycles in the graph $G^{'}$ given in algorithm $2$. Note that $G^{'}$ is copied into graph $G^{''}$ in Step $3$. In addition, note that set $S$ contains a set of cycles belonging to $G^{''}$ from which the cycle with the minimum $\it{Flow\,value}$ is found and deleted in steps $10$ and $11$, respectively. In order to prove $Lemma\,3.4$, assume the contradiction that algorithm $2$ does not delete cycle $C_{min}$ from the input graph $G^{'}$ which implies $C_{min} \notin S$. Let $C_1, C_2, C_3, \cdots, C_{k-1}, C_k$ be the cycles in set $S$ ordered in the increasing order of their maximum-$Valued$-edges, $i.e.$, $Value(M(C_{1})) < Value(M(C_{2})) < Value(M(C_{3})) < \cdots < Value(M(C_{k-1})) < Value(M(C_{k}))$, where $M(C)$ denotes the maximum-$Valued$-edge in a given cycle $C$. Note that every cycle in graph $G^{''}$ contains edge $e$, which is directed from vertex $u$ to vertex $v$, since algorithm $2$ is invoked by algorithm $1$ in Step $7$ when the addition of edge $e$ created a cycle in graph $G^{'}$. Now let us complete the proof of $Lemma\,3.4$ using the following $3$ exhaustive cases.

\begin{itemize}

\item $\it{\bf{Case\,1:}}$ $\bf{Value\big(M(C_{k})\big) < Value\big(M(C_{min})\big)}$

According to $Lemma\,3.3$, $C_k$ is the first cycle to be removed from graph $G^{''}$ in Step $7$ of algorithm $2$. By definition, $\it{Flow\,value}$ of any cycle $C$ = $\phi_{C} = \big(Value(M(C)) - Value(\gamma_{C})\big)$. Therefore, in $Case\,1$ where $Value\big(M(C_{min})\big) > Value\big(M(C_{k})\big)$, since $\phi_{C_{min}} < \phi_{C_{k}}$, $Value(\gamma_{C_{min}}) > Value(\gamma_{C_{k}})$. If we remove the common edge $e$ (which is directed from vertex $u$ to vertex $v$ in graph $G^{''}$) from both the cycles, $C_{min} - \{e\}$ and $C_{k} - \{e\}$ are now two paths directed from vertex $v$ to $u$ such that $Value(\gamma_{C_{min} - \{e\}}) > Value(\gamma_{C_{k} - \{e\}})$. As $G^{''}$ is the input graph given to invoke algorithm $3$, according to $Corollary\,3.2$ the path $C_{k} - \{e\}$ found by algorithm $3$ should be a $\it{Maxflow\,path}$ from $v$ to $u$. However, this is not the case as $Value(\gamma_{C_{min} - \{e\}}) > Value(\gamma_{C_{k} - \{e\}})$. Hence $\it{Case\,1}$ is not valid.

\item $\it{\bf{Case\,2:}}$ $\bf{Value\big(M(C_{min})\big) < Value\big(M(C_{1})\big)}$

After the removal of cycle $C_1$, cycle $C_{min}$ will be present in $G^{''}$, because, in Step $7$, when all edges in $G^{''}$ whose $Value \geq$ $Value\big(M(C_{1})\big)$ are deleted, none of the edges in $C_{min}$ are deleted as $Value\big(M(C_{min})\big)$ $<$ $Value\big(M(C_{1})\big)$. The presence of cycle $C_{min}$ implies the presence of the path $C_{min} - \{e\}$, which starts from vertex $v$ and ends in vertex $u$. According to $Lemma\,3.3$, $C_1$ is the last cycle to be found in Step $7$ of algorithm $2$ and this contradicts $Corollary\,3.2$ as there exist the path $C_{min} - \{e\}$. Hence $\it{Case\,2}$ is also invalid.

\item $\it{\bf{Case\,3:}}$ $\bf{Value\big(M(C_{i})\big) < Value\big(M(C_{min})\big) <}$ $\bf{Value\big(M(C_{i+1})\big)}$

This case can easily be proved by using the arguments given in $Case\,1$ and $Case\,2$. After the removal of cycle $C_{i+1}$ from graph $G^{''}$ in Step $7$ of algorithm $2$, cycle $C_{min}$ will still be present in $G^{''}$ since in Step $7$ when all edges in $G^{''}$ whose $Value \geq$ $Value\big(M(C_{i+1})\big)$ are deleted, none of the edges in cycle $C_{min}$ got deleted as $Value\big(M(C_{min})\big)$ $<$ $Value\big(M(C_{i+1})\big)$. Now recall that for any cycle $C$, $\phi_{C} = \big(Value(M(C)) - Value(\gamma_{C})\big)$. Therefore, in this case, $Value\big(M(C_{min})\big) > Value\big(M(C_{i})\big)$ and $\phi_{C_{min}} < \phi_{C_{i}}$ implies $Value(\gamma_{C_{min}}) > Value(\gamma_{C_{i}})$. If we remove the common edge $e$ (which is directed from vertex $u$ to vertex $v$) from both the cycles, $C_{min} - \{e\}$ and $C_{i} - \{e\}$ are now two paths directed from vertex $v$ to $u$ such that $Value(\gamma_{C_{min} - \{e\}}) > Value(\gamma_{C_{i} - \{e\}})$. According to $Lemma\,3.3$, $C_i$ is the cycle found by Step $7$ of algorithm $2$ after the deletion of cycle $C_{i+1}$, but this contradicts $Corollary\,3.2$ as the path $C_{i} - \{e\}$ found by algorithm $3$ is not a $\it{Maxflow\,path}$ from vertex $v$ to vertex $u$ since $Value(\gamma_{C_{min} - \{e\}}) > Value(\gamma_{C_{i} - \{e\}})$. Hence $\it{Case\,3}$ is not valid.

\end{itemize}

This completes the proof of $Lemma\,3.4$.
\end{proof}

\begin{theorem}
Algorithm $1$ produces a directed acyclic graph by deleting all cycles from the graph $G^{'}$ in which the cycle with the minimum $\it{Flow\,value}$ gets deleted before the other cycles.
\end{theorem}

\begin{proof}
In $Lemma\,3.4$, we have already proved that algorithm $2$ deletes a cycle with the minimum $\it{Flow\,value}$ among all the cycles in graph $G^{'}$. In addition, the while loop defined in Step $6$ of algorithm $1$ invokes algorithm $2$ (in Step $7$) until $G^{'}$ has no cycles left in it. This proves the theorem.
\end{proof}

\subsection{Algorithm analysis}

\begin{theorem}
If $n$ is the number of vertices and $m$ is the number of edges in the input graph $\vec{G}$ given to algorithm $1$, then, algorithm $1$ runs in $\mathcal{O}\big((m+n) \cdot m^2 \cdot log(n)\big)$ time. 
\end{theorem}

\begin{proof}
In algorithm $3$, if we are using a max heap for deleting the vertices with the largest dist$[]$ value in Q, then, in the worst case it runs in $\mathcal{O}\big((m+n) \cdot log(n)\big)$ time. In algorithm $2$, as one can easily observe, the while loop from steps $4-8$ takes the maximum amount of time. In the worst case, it may run Step $5$ for $\mathcal{O}(m)$ time. Hence, algorithm $2$ runs in $\mathcal{O}\big((m+n) \cdot m \cdot log(n)\big)$ time. Finally, in algorithm $1$, the while loop in steps $3-8$ may run in $\mathcal{O}(m)$ time in the worst case were the addition of edges in Step $5$ creates a cycle in almost all cases. Hence, the total time taken by algorithm $1$ is $\mathcal{O}\big((m+n) \cdot m^2 \cdot log(n)\big)$.
\end{proof}

  \begin{algorithm}
   \caption{Weighted Cycle Deletion}
    \begin{algorithmic}[1]
      \Procedure{WCD}{$\vec{G}=(V,\vec{E})$}
      \Statex
      \LeftComment{$\vec{G}$ is a weighted directed graph with multiple-edges and no self-loops} \Statex
      \LeftComment{Weight on each edge is a tuple with $Value$ and $Time$, $(Value,Time)$} \Statex
      \LeftComment{Edges of graph $\vec{G}$ are stored in edge-set $\vec{E}$ in their chronological order}
      \Statex

        \State $\bf{Initialize}$ $G^{'} = \emptyset$ \Statex \LeftComment{$G^{'}=(V^{'},E^{'})$, hence, $(G^{'} = \emptyset) \implies (V^{'} = E^{'} = \emptyset)$}
        \Statex
        
         \While{($\vec{E} \neq \emptyset$)} 
    \State $e = $ DEQUEUE($\vec{E}$)\Statex \LeftComment{Edge $e$ is the least recent edge} \Statex
    
    \State $E^{'} = E^{'} \cup e$ \Statex \LeftComment{Note that $V^{'}$ also gets updated in the process} \Statex
    
    \While{($G^{'}$ has a cycle)} \LeftComment{DFS is used here}
    \State $G^{'}$ = $\bf{function}$ DELETE\_CYCLE($G^{'},e$)
    \EndWhile
    \Statex
  \EndWhile \Statex
   \LeftComment{Graph $G^{'}$ now contains the desired DAG} \Statex
       \EndProcedure

\end{algorithmic}
\end{algorithm}

 \begin{algorithm}
   \caption{Function definition of DELETE\_CYCLE()}
    \begin{algorithmic}[1]
      \Function{DELETE\_CYCLE}{$G^{'},e$} \Statex
      \LeftComment{edge $e$ is the most recently added edge in $G^{'}$ that formed the cycle} \Statex
       \LeftComment{$Value(e)$ gives the $Value$ of edge $e$ from its ordered $2-$tuple $(Value,Time)$}

       \Statex
            
        \State Let vertex-tuple $(u,v)$ define the directed edge $e$
        \Statex
        \LeftComment{$i.e.$ edge $e$ is directed from vertex $u$ to vertex $v$} \Statex

        \State $\bf{Initialize}$ set $S = \emptyset$ and $G^{''} = G^{'}$
        \Statex \LeftComment{$G^{''}=(V^{''},E^{''})$, hence, $(G^{''} = G^{'}) \implies (V^{''} = V^{'}$ and $E^{''} = E^{'})$}
        \Statex
         \While{($G^{''}$ has a cycle)} \LeftComment{DFS is used here}
    \State $P$ = $\bf{function}$ MAX\_MIN($G^{''},u,v$)
    \Statex
    \LeftComment{$P$ contains the $\it{Maxflow\,path}$ $\mu_{vu} \in G^{''}$}
    \Statex
     
    \State $S = (S \cup P)$ \Statex \LeftComment{$S$ contains a set of ordered tuples, where each tuple denotes a path from $v$ to $u$}
    \Statex
    \State $\bf{Delete}$ edge $e^{''}\in E^{''},$ where, \Statex $Value(e^{''}) \geq Value(e_{max})$ \Statex \LeftComment{$e_{max}$ is the edge with the largest $Value$ in $P$}
    
   \EndWhile
   \Statex
    
    \State $\forall P^{'} \in S$, update $P^{'} = (P^{'} \cup \{e\})$ \Statex \LeftComment{Add edge $e$ to each of the ordered tuple in $S$} \Statex
    
    \State Find $P_{min} \in S$, such that, the $\it{Flow\,value}$ of $P_{min}$ is the minimum among all the cycles present in $S$ \Statex \Statex
    
    
    \LeftComment{$e_{min}$ be the minimum-valued-edge in $P_{min}$}
    \State $\bf{Delete}$ a flow of $Value(e_{min})$ from all the edges of $P_{min} \in G^{'}$ \LeftComment{$i.e.,$ $\forall e \in P_{min} \in G^{'}$, \Statex \hskip 7.35em $Value(e) = (Value(e) - Value(e_{min}))$}
    \Statex
        
    \State $\bf{Return}$ graph $G^{'}$ \Statex

\EndFunction
\end{algorithmic}
\end{algorithm}

  \begin{algorithm}
   \caption{Function definition of MAX\_MIN()}
    \begin{algorithmic}[1]
      \Function{MAX\_MIN}{$(G^{''},u,v)$}
      \Statex
      \LeftComment{Here we use two vectors mapped to each of the vertices in $V^{''}$, $viz.$, dist$[]$ and parent$[]$}
      \Statex
       \LeftComment{$Value(e)$ gives the $Value$ of edge $e$ from its ordered $2-$tuple $(Value,Time)$}
       \Statex
      
        \State $\forall w \in V^{''} \setminus v,$ $\bf{Initialize}$
        dist[w]= $-\infty$, parent$[w] = \emptyset$,
        \Statex \hskip1.7em dist$[v]= +\infty$, parent$[v] = \emptyset$ 
        \Statex
        \State Insert all vertices in $V^{''}$ to Queue $Q$ in decreasing order of their dist$[]$ values\Statex \LeftComment{$\forall x \in V^{''}$ ENQUEUE$(x,Q)$ in decreasing order of dist$[x]$}
        \Statex
        
         \While{($Q \neq \emptyset$)} \Statex
    \State $ver$ = DEQUEUE$(Q)$ \LeftComment{Delete $ver$ from $Q$, where $ver$ is the vertex with the largest dist$[]$ value in $Q$}
         \Statex
     \State Let set $N$ contains all outgoing-neighbors of $ver$
    \Statex
    \LeftComment{outgoing-neighbors of a vertex $v$ are all vertices to which $v$ has an outward directed edge} \Statex

    \State $\forall n \in N$, \Statex
    \hskip3.2em val = minimum( dist$[ver]$, $Value(e_{n})$ ) \Statex
    \LeftComment{$e_{n}$ is the edge with the highest $Value$ among all the edges directed from vertex $ver$ to vertex $n$ }\Statex
    \hskip3.2em $\bf{If}$ dist$[n]$ $<$ val
    $\bf{then}$ 
    \Statex \hskip4.5em dist$[n]$ $=$ val, parent$[n]$ = $e_{n}$

    \Statex
    
    
  \EndWhile
  \Statex
  
    \State $\bf{Return}$ the path $P$ from vertex $v$ to vertex $u$ \Statex
    \LeftComment{Path $P$ can be found by backtracking from the vertices present in parent$[u]$ to vertex $v$} \Statex
    
    \EndFunction

\end{algorithmic}
\end{algorithm}

\section{Case Study}
We had taken up a case in which eight dealers are doing intensive $\it{circular\,trading}$ among themselves. Figure $5$ shows the details of the same in the form of a directed graph with vertices denoting the dealers, and directed edges denoting the direction of transactions along with the total amount of tax paid (in lakh of \rupee, 1 lakh = \rupee~1,00,000) to the seller by the buyer. Note that the weight on each edge in the graph shows the total tax paid by a particular buyer to a particular seller (total tax is the sum of all the taxes involved in multiple transactions between them).

\begin{figure}[h!]
\begin{center}
  \caption{Input graph}
\includegraphics[scale=0.48]{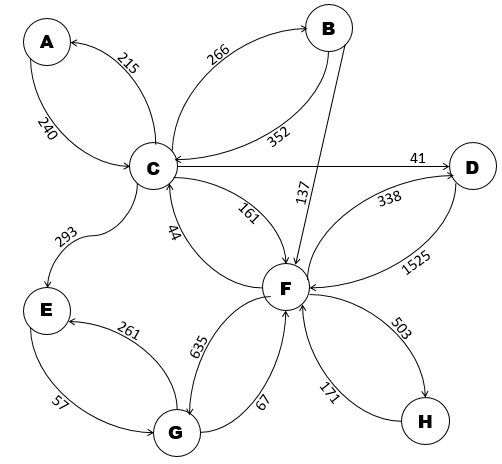}
\label{graph_table5}
\end{center}
\end{figure}

As given in Figure $5$, there are numerous cycles among the group of eight dealers and these cycles are considered as undesirable patterns by domain experts. Cycles are undesired in these transactions since a cycle indicate the buying of the same goods by a dealer that (s)he had previously sold. According to domain experts, this kind of flow of goods is not valid for the particular commodity that the eight people are trading among themselves. Hence, these cycles are indicators for tax-evasion. As mentioned in Section III, the dealers involved in circular trade fabricates the tax-values in such a way that the input tax and output tax due to the illegitimate transactions are almost the same. Hence we used the proposed algorithm to delete the illegitimate cycles. Without the removal of illegitimate cycles detecting the malicious transactions that are used to evade tax is almost impossible. Figure $6$ shows the directed acyclic graph obtained after deleting all cycles from the graph given in Figure $5$.

\begin{figure}[h!]
\begin{center}
  \caption{Malicious transactions}
\includegraphics[scale=0.45]{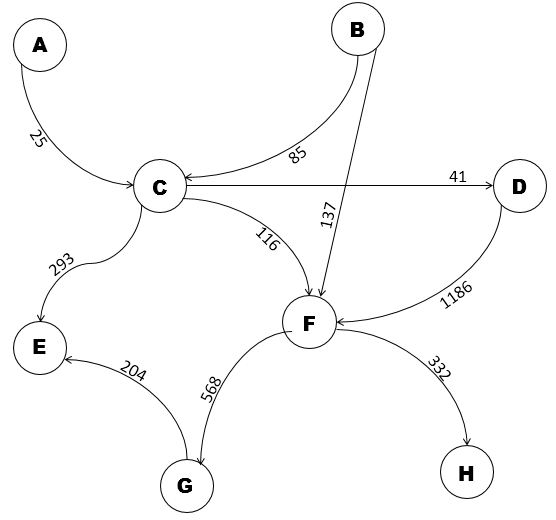}
\label{graph_table6}
\end{center}
\end{figure} 

The novelty of our technique over the others in deleting cycles is that we delete the cycle whose maximum-edge-$Value$ and minimum-edge-$Value$ are closest to each other before deleting the other cycles, and hence with a higher accuracy we delete the illegitimate edges that were used to form the cycle. On the other hand, if one doesn't take advantage of such patterns formed (as given in Observation $1$), which may be specific to the domain (which in our case is the Value Added Taxation system), the output DAG may not be of much use to the tax enforcement officers. As discussed before, the objective of the illegitimate transactions were to hide the malicious transactions, and the directed acyclic graph (DAG) given in Figure $6$ contains the required malicious transactions. Analysis of this DAG provides significant information, which can be used by the taxation authorities for conducting further investigations. The results obtained from their investigations are given in the following.

In their monthly tax return statements, all these dealers have shown huge purchases from certain dealers who are from outside the state and also not belong to this group of eight dealers. 

\begin{itemize}
\item	The eight dealers did total purchases of \rupee~798 crores, out of which non-creditable purchases (purchases from outside the state or international imports) are \rupee~622 crores. 
\item They paid \rupee~4.47 crores as VAT \& interstate sales tax (also known as CST).


\item	They have shown branch transfers(located branches in other states) of \rupee~230 crores on which no tax is required to be paid.
\item	They have shown questionable exports of \rupee~105 crores on which no tax is required to be paid.
\item	They have shown questionable inter state(CST) sales amounting to \rupee~111 crores on which a lesser rate of tax (@2\%) is applicable.
\item They have shown local VAT sales of \rupee~233 crores on which they have paid a tax of \rupee~2.47 crores.
\end{itemize}

By studying their whole purchases information from the malicious transactions (given in Figure $6$) and outside the state transactions, the taxation authorities observed that they should have paid \rupee~31.10 crores in tax. However, they actually paid only \rupee~4.47 crores as VAT \& CST by showing fictitious exports and inter state sales. Hence they evaded the payment of more than 85\% of tax.

From the original dataset, we observed that the number of dealers involved in any given $\it{circular\,trade}$ ranges from $2$ to $8$. This small number is due to the high risk factors involved in the process. Recall from $Step\,2$ of Sub-section $A$ in Section I that some of the dealers involved in a $\it{circular\,trade}$ are fictitious dealers and obtaining identification details for a large set of such dealers is a highly risky bargain. The dealers have to file their tax-returns on a monthly basis. In most of the $\it{circular\,trading}$ cases we encountered, the total number of transactions per month amount to several hundreds and rarely few thousands. In order to give an idea about the time taken by the proposed algorithm, we run it for synthetic data-sets of varying input size. For the same, we used a machine with an $8$ GB RAM and a $2.20$ GHz Intel Core $i5$ processor. Figure $7$ contains the graph that describe the time taken in seconds (on $Y$-axis) for varying input-graph sizes in terms of number of edges (on $X$-axis).

\begin{figure}[h!]
\begin{center}
  \caption{Running time for synthetic data-sets}
\includegraphics[scale=0.435
]{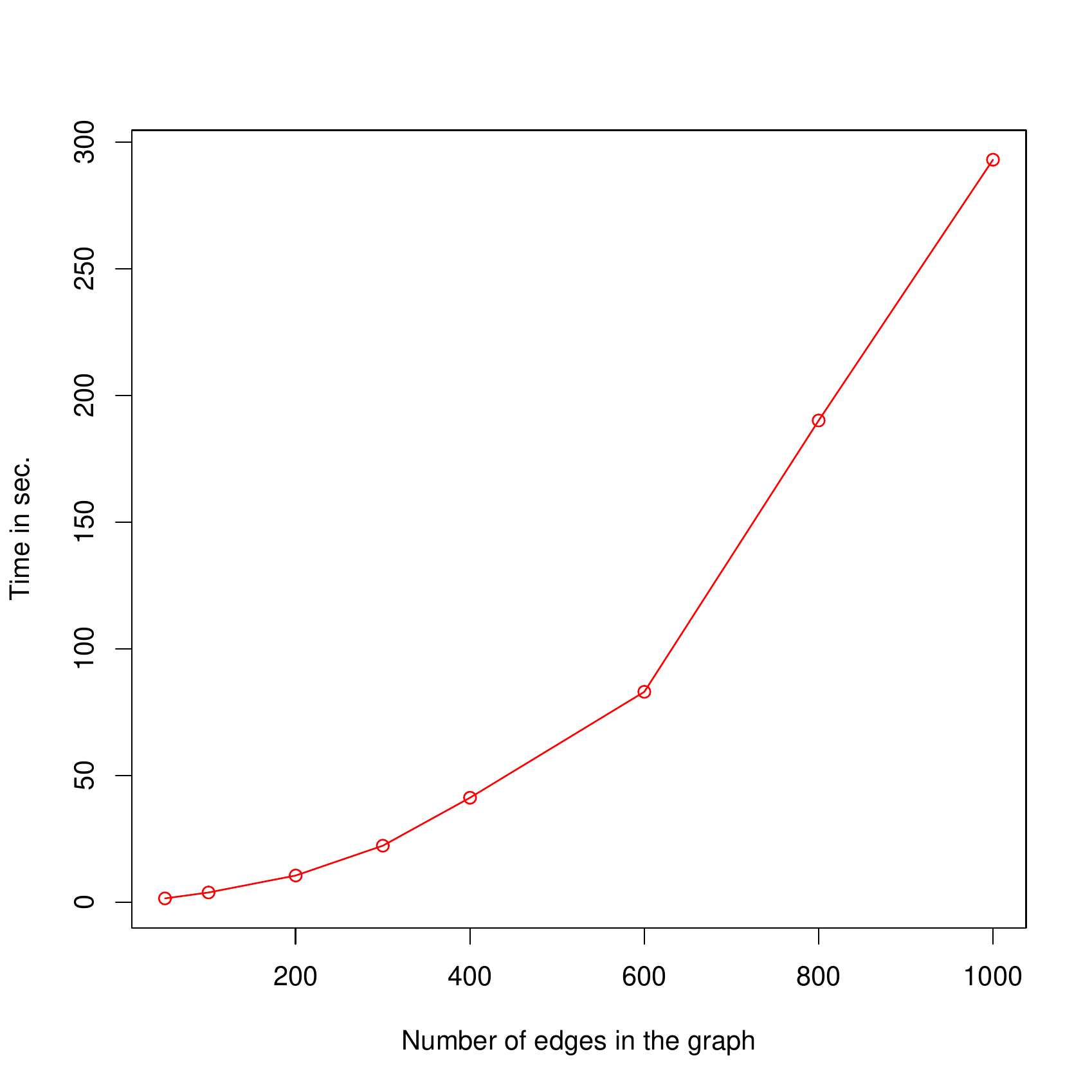}
\label{graph_table7}
\end{center}
\end{figure}

\/*

\section{Case Study}
We analyzed a case in which eight dealers are doing intensive $\it{circular\,trading}$ among themselves. Figure $5$ shows the details of the same in the form of a directed graph with vertices denoting the dealers, and directed edges denoting the direction of transactions along with the total amount of tax paid (in lakh of \rupee, 1 lakh = \rupee~1,00,000) to the seller by the buyer.

\begin{figure}[h!]
\begin{center}
  \caption{A known case of $\it{circular\,trading}$}
\includegraphics[scale=0.6]{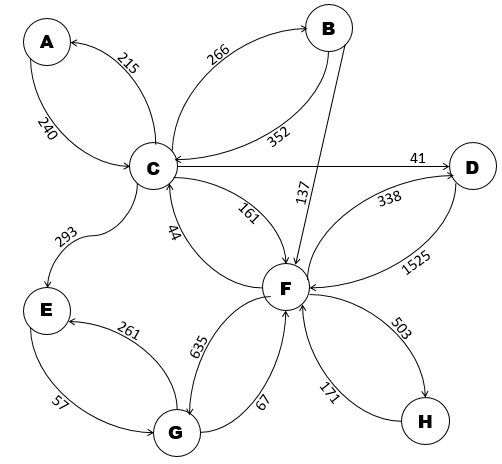}
\label{graph_table8}
\end{center}
\end{figure}

In their monthly tax return statements, all the eight dealers show huge purchases from outside the state. Legally, they should have paid heavy taxes on all these purchases. The following points illustrate a brief overview of the transactions among them.

\begin{itemize}
\item	The eight dealers did total purchases of \rupee~798 crores, out of which non-creditable purchases (purchases from outside the state or international imports) are \rupee~622 crores. 
\item	They should have paid a total tax of \rupee~31.10 crores, but they paid only \rupee~4.47 crores as VAT \& interstate sales tax (also known as CST).
\item	Hence, they evaded the payment of about 85\% of tax.
\end{itemize}

They have done this by using the following ways:-

\begin{itemize}
\item	Most of the dealers have shown branch transfers (branches located in other states) which amounts to a total of \rupee~230 crores on which no tax is required to be paid.
\item	They have shown questionable amount of exports totalling to \rupee~105 crores on which no tax is required to be paid.
\item	They have shown questionable amount of interstate(CST) sales totalling to \rupee~111 crores on which a much lesser rate of tax (@2\%) is applicable.
\item	They have also shown local VAT sales of \rupee~233 crores in total on which the output tax is \rupee~11.65 crores, but have paid only \rupee~2.47 crores to the government. They could do this by raising invoices among the group members and showing Input Tax Credit (ITC). This is where $\it{circular\,trading}$ comes into picture.
\end{itemize}

Figure $6$ shows the directed acyclic graph obtained after deleting all cycles from the graph given in Figure $5$ using the algorithms described above. Note that the weight on each edge in the graph given in Figure $6$ shows the total tax paid by a particular buyer to a particular seller (total tax is the sum of all the tax values involved in multiple transactions between them). For example, as one can observe in the edge from vertex $A$ to vertex $C$ in Figure $6$, the total tax involved between them after deleting many transactions to remove the cycle given in Figure $5$ using the proposed algorithm is \rupee~$25$ lakhs. It is important to note that, here the set of transactions that makes up the sum of \rupee~25 lakhs is the point of interest to tax authorities.

\begin{figure}[h!]
\begin{center}
  \caption{The output DAG}
\includegraphics[scale=0.5]{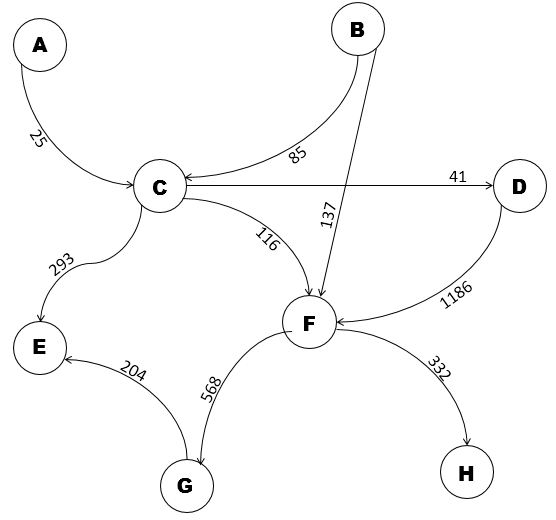}
\label{graph_table9}
\end{center}
\end{figure}

*/

\section{Conclusion}
In this paper, we formalized the infamous tax-evasion technique called $\it{circular\,trading}$. In $\it{circular\,trading}$, a group of traders fabricates heavy sales and(or) purchase transactions among themselves, which results in the flow of goods in a circular manner without any value addition, $ie.$, the input tax and the output tax due to the illegitimate transactions remains the same. The motivation to create such illegitimate transactions is to hide the malicious tax return information they have submitted to the tax authorities. The problem of removing the hence formed cycles is important as the tax authorities can easily detect the malicious transactions once the cycles are removed. Here, we proposed an algorithm to remove such cycles by making use of an important observation that the amount of tax payable by a dealer due to illegitimate sales and purchases transactions is almost zero. In addition, we run the proposed algorithm on a nexus of dealers involved in $\it{circular\,trade}$ and the results obtained are specified. In future, we try to define centrality measures for detecting the key players in $\it{circular\,trading}$. In addition, we plan to investigate whether there are more effective ways for removing cycles.


\section{Acknowledgment}

We are very grateful to the Telangana state government, India, for sharing the commercial tax dataset, which is used in the proposed work.
This work has been supported by Visvesvaraya PhD Scheme for Electronics and IT, Media Lab Asia, grant number EE/2015-16/023/MLB/MZAK/0176.

\ifCLASSOPTIONcaptionsoff
  \newpage
\fi

\end{document}